\documentclass[11pt]{article}
 \usepackage{mdframed}
\usepackage{hyperref}

\newcommand{\norm}[1]{\left\lVert#1\right\rVert}
\newcommand{\abs}[1]{\left\lvert#1\right\rvert}
\usepackage[margin=1in]{geometry}

\usepackage{amsfonts,amsthm,amssymb}
\usepackage{amsmath,amsopn}

\usepackage[colorinlistoftodos]{todonotes}

\DeclareMathOperator*{\argmax}{argmax}
\DeclareMathOperator*{\argmin}{argmin}

\newcommand{\functionname}[1]{\text{\sf #1}}
\newcommand{\fhtw}{\functionname{fhtw}}
\newcommand{\tw}{\functionname{tw}}

\newcommand{\calH}{\mathcal H}
\newcommand{\calV}{\mathcal V}
\newcommand{\calE}{\mathcal E}

\newcommand{\be}{\begin{enumerate}}
\newcommand{\ee}{\end{enumerate}}
\newcommand{\bi}{\begin{itemize}}
\newcommand{\ei}{\end{itemize}}
\newcommand{\beq}{\begin{equation}}
\newcommand{\eeq}{\end{equation}}

\newcommand{\bp}{\begin{proof}}
\newcommand{\ep}{\end{proof}}
\newcommand{\bcor}{\begin{cor}}
\newcommand{\ecor}{\end{cor}}
\newcommand{\bthm}{\begin{thm}}
\newcommand{\ethm}{\end{thm}}
\newcommand{\blmm}{\begin{lmm}}
\newcommand{\elmm}{\end{lmm}}
\newcommand{\bdefn}{\begin{defn}}
\newcommand{\edefn}{\end{defn}}
\newcommand{\bprop}{\begin{prop}}
\newcommand{\eprop}{\end{prop}}
\newcommand{\bconj}{\begin{conj}}
\newcommand{\econj}{\end{conj}}
\newcommand{\bopm}{\begin{opm}}
\newcommand{\eopm}{\end{opm}}
\newcommand{\brmk}{\begin{rmk}}
\newcommand{\ermk}{\end{rmk}}

\newcommand{\suchthat}{\ | \ }

\theoremstyle{plain}                   
\newtheorem{theorem}{Theorem}
\newtheorem{lemma}[theorem]{Lemma}
\newtheorem{definition}[theorem]{Definition}
\newtheorem{corollary}[theorem]{Corollary}

\theoremstyle{definition}              

\newtheorem{opm}{Open Problem}
\newtheorem{conj}{Conjecture}

\newtheorem{defn}{Definition}

\newtheorem{rmk}{Remark}

\title{A Relational Gradient Descent Algorithm For Support Vector Machine Training} 

 \author{Mahmoud Abo-Khamis \\ Relational AI
\and Sungjin Im~\thanks{Supported in part by NSF grants CCF-1409130, CCF-1617653, and CCF-1844939.} \\ University of California, Merced \and Benjamin Moseley~\thanks{Supported in part by NSF grants CCF-1725543, 1733873, 1845146, a Google Research Award, a Bosch junior faculty chair and an Infor faculty award.} \\ Carnegie Mellon University \and Kirk Pruhs~\thanks{Supported in part by NSF grants CCF-1421508 and CCF-1535755, and an IBM Faculty Award.} \\ University of Pittsburgh \and Alireza Samadian \\ University of Pittsburgh}

\begin{document}

\maketitle


\begin{abstract}
We consider gradient descent like algorithms for
Support Vector Machine (SVM) training 
when the data is in relational form.
The gradient of the SVM objective can not
be efficiently computed by known techniques 
as it suffers from the ``subtraction problem''.
We first show that the subtraction problem
can not be surmounted by showing
that computing any constant approximation
of the gradient of the SVM objective function is $\#P$-hard, even for acyclic joins. 
We however circumvent the subtraction problem by
restricting our attention to stable instances,
which intuitively are instances
where a nearly optimal solution remains nearly optimal  if the points are perturbed slightly.
We give an efficient algorithm
that computes a ``pseudo-gradient'' that guarantees
convergence for stable instances at a rate comparable to
that achieved by  using the actual gradient.
We believe that our results suggest
that this sort of stability 
 analysis would likely yield useful insight in context of designing  algorithms on relational data for other learning problems in which the subtraction problem arises. 

\end{abstract}

\section{Introduction}

Kaggle  surveys~\cite{KaggleSurvey} show that 
the majority of learning tasks faced by  data scientists involve relational data. Most commonly the relational data is stored in tables in a relational database. So these data scientists want to compute something like
\begin{quote}
\textbf {Data Science Query} =  Standard\_Learning\_Task(Relational Tables $T_1, \ldots T_m$)
\end{quote}
However, almost all standard algorithms for standard
learning problems 
assume that the input consists of points in Euclidean space~\cite{HundredPage},
and thus 
are not designed to  operate directly on relational data. 
The current standard practice for a data scientist, confronted with a learning task on relational data, is: 

\begin{enumerate}
\item
Firstly, convert any nonnumeric categorical data to numeric data.  As there are standard methods to accomplish this~\cite{HundredPage}, and
as we do not innovate with respect to this process, we will assume that all data is a priori numerical, so we need not  consider this step. 
\item
Secondly, issue a feature extraction query to extract the  data from the relational database by joining together the tables to materialize a design matrix $J=T_1 \Join \dots \Join T_m$ with say $N$ rows and $(d+1)$ columns.  Each row of this design matrix is then interpreted as a point in $d$-dimensional Euclidean space with an associated label. 
\item
Finally this design matrix $J$ is important into a  standard learning algorithm to train the model.
\end{enumerate}
Thus conceptually, standard  practice transforms a  data science query to a query of the following form:
\begin{quote}
\textbf {Data Science Query} =  Standard\_Learning\_Algorithm(Design Matrix $J = T_1 \Join \dots \Join T_m$)
\end{quote}
where the joins are evaluated first, and the learning algorithm is then applied to the result. 
Note that if each table has $n$ rows, the design
matrix $J$ can have as many as $n^m$ entries.  
Thus, independent of the learning task, this standard practice necessarily has exponential worst-case time and space complexity as the design matrix can be exponentially larger than the underlying relational tables. 
Thus a natural research question is what we call the relational learning question:
\begin{mdframed}
\noindent
\textbf{The Relational Learning Question:}
\renewcommand{\labelenumi}{\Alph{enumi}.}
\begin{enumerate}
    \item 
Which standard learning algorithms can  be implemented as relational algorithms, which informally are algorithms that
are efficient when the input is in relational form?
\item
And for those standard algorithms that
are not implementable by a relational algorithms, is there an alternative
relational algorithm that has the same performance guarantee as the standard algorithm?
\item
And if we can't find an alternative
relational algorithm that has the same performance guarantees to the standard algorithm, is there an alternative
relational algorithm that has some reasonable performance guarantee (ideally similar to the performance guarantee for the standard algorithm)?
\end{enumerate}
\end{mdframed}
Note that a relational algorithm can not afford to explicitly join the relational tables.

One immediate difficulty that we run into is that
if the tables have a sufficiently complicated structure, almost all natural problems/questions
about the design matrix are NP-hard if the data is in relational form. For example, it is NP-hard to even determine whether or not the design matrix is empty or not (see for example \cite{grohe2006structure,marx2013tractable}). 
Thus, as we want to focus on the complexity of the learning problems, we conceptually want to   abstract out the complexity 
of the tables. The simplest way to accomplish this is to  primarily focus on
 instances where the structure of the tables is
simple, with the most natural candidate for ``simplicity'' being that the join is acyclic. Acyclic joins are the norm in practice, and are a commonly considered special case in the database literature. 
For example,  there are efficient algorithms to compute the size of the design matrix for acyclic joins.  

Formally defining what an ``relational'' algorithm is
problematic, as for each  natural candidate definition there
are plausible scenarios in which that candidate definition is not the ``right'' definition. But for the purposes of this paper
it is sufficient to think of a ``relational'' algorithm as one whose
runtime is polynomially bounded in $n$, $m$ and $d$ if the join is acyclic. 

\noindent
{\bf Our Research Question:}
In this paper we address the relational learning question 
within the context of gradient descent algorithms for the
classic (soft-margin linear) Support Vector Machine (SVM) training problem.
SVM is identified as one of the five most important learning
problems in \cite{HundredPage}, and is covered in almost
all introductory machine learning textbooks. Gradient descent is probably the most commonly used computational technique for solving convex
learning optimization problems~\cite{sra2012optimization}.
So plan A is to find a relational implementation of gradient
descent for the SVM objective.
And if plan A fails, plan B is to find a 
 relational descent algorithm that has the same performance guarantee as gradient descent.
And finally, if both plan A fail and plan B fail, plan C is to find a
 relational algorithm that  has some other reasonable performance guarantee.

\subsection{Background}

We now give the minimal background on gradient descent and
SVM required to understand our results. 

\noindent{\bf Gradient Descent:} Gradient descent is a first-order iterative optimization method for finding an approximate minimum of a convex function $F: \mathbb{R}^d \to \mathbb{R}$, perhaps subject to a constraint the solution lies in some convex body $\mathcal{K}$. 
In the $G$ descent algorithm, at each descent step $t$ 
the current candidate solution $\beta^{(t)}$ is updated according to the following rule:
\begin{align}
   \beta^{(t)} \gets \beta^{(t-1)} - \eta_t G(\beta^{(t-1)})
\end{align}
where $\eta_t$ is the step size.
In projected $G$ descent, 
 the current candidate solution $\beta^{(t)}$ is updated according to
 the following rule:
\begin{align}
   \beta^{(t)} \gets \Pi_{\mathcal{K}} \left( \beta^{(t-1)} - \eta_t G(\beta^{(t-1)}) \right)
\end{align}
where $\Pi_{\mathcal{K}}(\alpha) =  \argmin_{\beta \in \mathcal{K}} \norm{\alpha - \beta}_2$ is the projection of
the point $\alpha$ to the closest point to $\alpha$ in $\mathcal{K}$. 
In (projected) gradient descent, $G$ is $\nabla F(\beta^{(t)})$, the gradient of $F$ at $\beta^{(t)}$.
 There are lots of variations of gradient descent,
 including variations on the step size, and variations, 
 like stochastic gradient descent~\cite{sra2012optimization}, in which the
 gradient is only approximated.

\noindent{\bf SVM training:}
Conceptually the input to SVM training consists of a collection $X = \{x_1,x_2,\dots,x_N\}$ of  points in $\mathbb{R}^d$, and a collection $Y = \{y_1,y_2,\dots,y_N\}$ of  associated labels from $ \{-1,1\}$.
For convenience let us rescale the points
 so that each point in $X$ lies within the hypercube $[-1, 1]^d$. 
A feasible solution is a $d$-dimensional vector $\beta$,
sometimes called a hypothesis.
The objective is to minimize a linear combination $F(\beta, X, Y)$ of the average ``hinge'' loss
function of the points $L(\beta, X, Y)= 	\frac{1}{N}\sum_{x_i \in X} \max(0, 1-y_i\beta x_i)$ plus a regularizer $R(\beta)$.
We will take the regularizer to be the 2-norm squared of $\beta$, as that is a standard choice~\cite{HundredPage}, although this choice is not so important for our purposes. 
Thus the objective is to minimize:
\begin{align}
 F( \beta, X, Y) = 	\frac{1}{N}\sum_{x_i \in X} \max(0, 1-y_i \beta x_i) + \lambda ||\beta||_2^2 
\end{align}
Here the loss function measures how well the hypothesis $\beta$ explains the labels, and one of the regularizer's purposes is to prevent overfitting. The $\lambda$ factor intuitively specifies the amount that the loss has to decrease to justify an increase in the norm of $\beta$.
When either $X$ and $Y$ is understood, for notational convenience, we may drop them from the objective.

\noindent{\bf Gradient Descent for SMV:}
In Section \ref{sect:gradientdescentanalysis}
we show that by a straightforward specicialization of a standard convergence analysis for projected gradient
descent to SVM one obtains Theorem \ref{corollary:gradientsvm},
which bounds the number of
descent steps needed to reach a solution with a specified
relative error.

\begin{theorem}
\label{corollary:gradientsvm}
Let  $F(\beta)$ be the SVM objective function.
Let $\beta^* = \argmin_\beta  F(\beta)$ be the optimal solution.
Let $\widehat{\beta}_s = \frac{1}{s}\sum_{t=0}^{s-1} \beta^{(t)}$. 
Let $\eta_t = \frac{1}{8 \lambda \sqrt{dt}}$. Then
if $ T \ge \left( \frac{4  d^{3/2}     }{\epsilon \lambda F(\widehat{\beta}_{T})} \right)^2$ then projected gradient descent guarantees
that 
\begin{align*}
    F(\widehat{\beta}_{T}) \leq  (1 +\epsilon) F(\beta^*)  
\end{align*}
Thus if the algorithm returns $\widehat \beta$ at the
first time $t$ where $ t \ge \left( \frac{4  d^{3/2}     }{\epsilon \lambda F(\widehat{\beta}_{t})} \right)^2$, then it
achieves relative error at most $\epsilon$.
\end{theorem}

\subsection{Our Results}

We start by making some observations about  the gradient
\begin{align}
\nabla F = 2 \lambda \beta - \frac{1}{N} \sum_{i \in \mathcal{L}} y_i x_i
\end{align}
of the SVM objective function $F$. 
First note the term $2 \lambda \beta$ is trivial to compute,
so let us focus on the term $G=\frac{1}{N} \sum_{i \in \mathcal{L}} y_i x_i$.
Firstly only those points $x_i$ that satisfy the additive
constraint $\mathcal{L}$ contribute to the gradient. 
Now let us focus on a particular dimension, and use
$x_{ik}$ to refer to the value of point $x_i$ in dimension $k$. 
Let $L_k^- = \{ i \mid i \in \mathcal{L} \text{ and } y_i x_{ik} < 0 \}$ denote those points that satisfy $\mathcal{L}$ 
and whose the gradient in the $k^{th}$ coordinate has negative sign.
Conceptually each point in $L_k^-$ pushes the gradient in dimension $k$ up with ``force''
proportional to its value in dimension $k$.
Let $L_k^+ = \{ i \mid x_i \in \mathcal{L} \text{ and } y_i x_{ik} > 0 \}$ denote those points that satisfy $\mathcal{L}$ 
and whose the gradient in the $k^{th}$ coordinate has positive sign.
And conceptually each point in $L_k^+$ pushes the gradient  in dimension $k$ down  with ``force''
proportional to its value in dimension $k$.

 Next we note that $ G = \frac{1}{N} \sum_{i \in \mathcal{L}} y_i x_i$ 
 is what is called a FAQ-AI(1) query in~\cite{FAQ-AI(1),faqai}. 
 \cite{FAQ-AI(1)} gives a relational approximation scheme (RAS) for
 certain  FAQ-AI(1)  queries.
 A RAS is a collection $\{ A_\epsilon\}$ of relational
 algorithms where $A_\epsilon$ achieve $(1+\epsilon)$-approximation. The results in \cite{FAQ-AI(1)} can be applied
 to obtain a RAS to compute a $(1+\epsilon)$ approximation $\widehat G_k^+$ to $G_k^+ = \frac{1}{N} \sum_{i \in L_k^+} y_i x_{ik}$,
  and a RAS to compute
 a $(1+\epsilon)$ approximation $\widehat G_k^-$ to 
  $G_k^- =  \frac{1}{N} \sum_{i \in L_k^-} y_i x_{ik}$. 
 However, the results in \cite{FAQ-AI(1)}
can not be applied to get a RAS for computing 
a $(1+\epsilon)$-approximation to 
$G= G_k^- + G_k^+$,
as it suffers from what \cite{FAQ-AI(1)} calls the \emph{subtraction problem}. Conceptually the subtraction problem is
the fact that good approximations of scalars $a$ and $b$
are generally insufficient to deduce a good approximation
of $a - b$. This subtraction problem commonly arises
in natural problems, and several examples are given in \cite{FAQ-AI(1)}. 
Thus an additional reason for our interest in relational algorithms to compute the (perhaps approximate) 
gradient of the SVM objective function is that we
want to use it as test
case to see if there is some way that we can 
surmount/circumvent the subtraction problem, and obtain
a relational algorithm with a reasonable performance
guarantee, ideally using techniques that are applicable to
other problems in which this subtraction problem arises.

We start with a rather discouraging negative result that shows that we can not surmount
the subtraction problem in the context of
computing the gradient of the SVM objective problem.
In particular, we show in Section \ref{sec:gradient-hard} that  computing   an
$O(1)$ approximation to the partial derivative in a specified specified dimension
 is $\#P$-hard, even for acyclic joins. 
 This kills plan A as a relational algorithm to compute the
 gradient would imply $P = \#P$. This also makes
it hard to imagine plan B working out since, assuming $P \ne \#P$, a relational algorithm can't even be sure that it is  even approximately headed in the direction of the optimal solution, and thus its not reasonable to expect that we could find a relational algorithm to compute some sort of ``pseudo-gradient'' that would guarantee convergence on all instances.

Thus it seems we have no choice but to fall back to plan C.
That is, we have to try to circumvent (not surmount) the subtraction problem. After some reflection, one reasonable
interpretation of our $\#P$-hardness proof is that it
shows that computing the gradient is hard on unstable instances.
In this context, intuitively an instance is stable
if a nearly optimal solution remains nearly optimal  if the points are perturbed slightly.
Intuitively one would expect real world instances, 
where there is a hypothesis $\beta$ that
explains the labels reasonably well, to
be relatively stable (some discussion of the stability
of SVM instances can be found in \cite{bi2005support}).
And for instances where there isn't a hypothesis
that explains the labels reasonably well, 
it probably doesn't matter what
hypothesis the algorithm returns, as
it will likely be discarded by the data scientist anyways. 
Thus,  our plan C will be to seek a gradient descent algorithm that has a similar convergence 
guarantee to  gradient descent on stable instances. 

Long story short, the main result of this paper is that 
this plan C works out. That is we give a relational algorithm
that computes a ``pseudo-gradient'' that guarantees
convergence for stable instances at a rate comparable to 
that achieved by  using the actual gradient. 
The algorithm design can be found in Section
\ref{sect:pseudogradient}, and the algorithm analysis can
be found in Section
\ref{sect:analysis}.
Postponing for the moment our formal definition
of stability, we state our
main result in Theorem \ref{thm:main}.
The reader should compare Theorem \ref{thm:main}
 to the analysis
of gradient descent in Theorem \ref{corollary:gradientsvm}.

\begin{theorem}
\label{thm:main}
Let $X$ be an $(\alpha, \delta, \gamma)$-stable SVM instance
formed by an acyclic join. 
Let $\beta^* = \argmin_\beta  F(\beta)$ be the optimal solution.
Then there is a relational algorithm that can
compute a pseudo-gradient in time
$O(\frac{m}{\epsilon^2}(m^3 \log^2(n))^2  (d^2 m n \log(n)) )$, where $\epsilon=\min(\frac{\delta}{8}, \alpha)$. 
 After $T=\left( \frac{256 d^{3/2} }{\lambda \delta F(\beta_a, X_a)} \right)^2$ iterations of  projected  descent using this pseudo-gradient there is a relational algorithm that
 can compute in time $O(\frac{1}{\epsilon^2}(m^3 \log^2(n))^2  (d^2 m n \log(n)))$ a hypothesis $\widehat \beta$ such that: 
\begin{align*}
    F(\widehat \beta, X) \leq (1 + \gamma) F(\beta^*,X) 
\end{align*}
\end{theorem}

\begin{mdframed}
\noindent
\textbf{Main Takeaway Point:}
In a broader context, we believe that our results suggest
that this sort of stability 
 analysis would likely yield useful insight in context of designing relational algorithms for other learning problems in which the subtraction problem arises. 
 \end{mdframed}

\subsection{Related Results}

Relational algorithms are known for certain types
of Sum of Sums (SumSum) and Sum of Products (SumProd) queries. In particular the Inside-Out algorithm~\cite{faq} can evaluate a SumProd query in time $O(m d^2 n^{h} \log n)$, where  $m$ is the number of tables, $d$ is the number of columns, and $h$ is the fractional hypertree width~\cite{GM06} of the query. 
Note that $h=1$ for the acyclic joins, and 
thus Inside-Out is a polynomial time algorithm for acyclic joins.
One can reduce SumSum queries to $m$ SumProd queries~\cite{faqai}, and thus they be solved in 
time $O(m^2 d^2 n^{h} \log n)$.
The Inside-Out algorithm builds on several earlier papers, including~\cite{Aji:2006,Dechter:1996,Kohlas:2008,GM06}.

SumSum and SumProd queries with additive inequalites 
was fist studied  in \cite{faqai}. \cite{faqai} gave an algorithm with worst-case time
complexity $O(m d^2 n^{m/2} \log n)$. So this is
better than the standard practice of forming the design matrix, which has worst-case time complexity $\Omega(d n^m)$. Different flavors of queries with inequalities were also studied~\cite{Klug:1988:CQC:42267.42273,Koutris2017,DBLP:journals/corr/abs-1712-07445}.
\cite{FAQ-AI(1)} showed that computing even very simple
types of SumSum and SumProd queries with a single
inequality is NP-hard. 
But an RAS for special types queries   is introduced in \cite{FAQ-AI(1)}. The algorithm in \cite{FAQ-AI(1)} can obtain $(1 + \epsilon)$ approximation for problems such as counting the number of rows on one side of a hyperplane in time $O(\frac{1}{\epsilon^2}(m^3 \log^2(n))^2  (d^2 m n^{h} \log(n)) )$.

Algorithms for 
 linear/polynomial regression on relational data are given in~\cite{SystemF,khamis2018ac,IndatabaseLinearRegression,Kumar:2015:LGL:2723372.2723713,Kumar:2016:JJT:2882903.2882952} and an 
 algorithm for $k$-means clustering on relational data is given in~\cite{Rkmeans}.

Stability analysis, similar in spirit to our results, has been consider before in clustering
problems~\cite{Roughgarden19,bilu2012stable,makarychev2014bilu,ackerman2009clusterability,balcan2013clustering,daniely2012clustering,kumar2010clustering,ostrovsky2013effectiveness,angelidakis2017algorithms,awasthi2012center,balcan2016clustering}. 
For example, the $NP$-hard $k$-means, $k$-medians and $k$-centers
clustering problems are polynomially solvable 
for instances in which changing the distances of the points by a multiplicative factor of at most $2$ does not change the optimal solution~\cite{angelidakis2017algorithms,awasthi2012center,balcan2016clustering}

SVM is discuss in covered in almost every introductory machine
learning textbook, for example \cite{HundredPage}.
Optimization methods for learning problems, including variations of gradient descent,
are discussed in \cite{sra2012optimization}. A overview of online 
convex optimization, which we use in our results, can be found in \cite{OCO,Hazan16}.

\section{Hardness of Gradient Approximation}
    \label{sec:gradient-hard}

\begin{lemma}
\label{lemma:NP-hardSVM}
It $\#P$ hard to  $O(1)$-approximate the
partial derivative of the SVM objective function in a specified 
dimension. 
\end{lemma}

\begin{proof}
We reduce the decision version of the counting knapsack problem to the problem of approximating the gradient of SVM. The input to the decision counting knapsack problem is a set of weights $W=\{w_1, w_2, \dots, w_m\}$, a knapsack size $L$, and an integer $k$. The output of the problem is whether there are $k$ different combinations of the items that fit into the knapsack.

We create $m+1$ tables, each with two columns. The columns of the first $m$ table are $(\text{Key}, E_i)$ for $T_i$ and the rows are $$T_i = \{(1,0) , (1,w_i/L), (0,0)\}.$$ The last table has two columns $(\text{Key}, \text{Value})$, and it has two rows $(1,1),(0,-k)$. Note that if we take the join of these tables, there will be $m+2$ columns: $(\text{Key}, \text{Value}, E_1, E_2, \dots, E_m)$. 

Let $\beta=(0,0,1,1,\dots,1)$ and $\lambda=0$, so $\beta$ is $0$ on the columns Key and Value and $1$ everywhere else. Then we claim, if the gradient of $F$ on the second dimension (Value) is non-negative, then the answer to the original counting knapsack is true, otherwise, it is false.

To see the reason, consider the rows in $J$: there are $2^m$ rows in the design matrix that have $(1,1)$ in the first two dimensions and all possible combinations of the knapsack items in the other dimensions. More precisely, the concatenation of $(1, 1)$ and $w_S$ for every $S \in [m]$ where $w_S$ is the vector that has 
$w_i / L$ in the $i$-th entry if item $i$ is in $S$ or 0 otherwise. Further, $J$ has a single special row with  values $(0,-k,0,0, \dots, 0)$. 
Letting $G_2$ be the gradient of SVM on the second dimension (column Value), we have,
\begin{align*}
    G_2 = \sum_{x \in J: 1-\beta x \geq 0} x_2
\end{align*}

For the row with $\text{Key}=1$ for each $S \in [m]$, we have $1-\beta x = 1 - \sum_{i \in S} w_i / L \geq 0$ if and only if the items in $S$ fits into the knapsack and $x_2 = 1$. For the single row with $\text{Key}=0$, we have $1 - \beta x = 1$, and its value on the second dimension is $x_2 = k$. Therefore, 
\begin{align*}
    G_2 = C_L(w_1,\dots,w_m) - k
\end{align*}
where $C_L$ is the number of subsets of items fitting into the knapsack of size $L$. This means if we could approximate the gradient up to any constant factor, we would be able to determine if $G_2$ is positive or negative, and as a result we would be able to answer the (decision version of) counting knapsack problem, which is $\#P$-hard.
\end{proof}

\section{Algorithm Design}
\label{sect:pseudogradient}

\subsection{Review of Row Counting with a Single Additive Constraint}

We now summarize algorithmic results from \cite{FAQ-AI(1)} for two different problems,
that we will use as a black box.

In the first problem the
input is a collection $T_1, \ldots, T_m$ of
tables, a label $\ell \in \{-1, +1\}$, and an additive inequality
$\mathcal{L}$ of the form $ \sum_{ j \in [d]} g_j(x_{j}) \geq R$, where each function
$g_j$ can be computed in constant time. 
The output consists of, for each $j \in [d]$ and $e \in D(j)$,
where $D(j)$ is the domain of column/feature $j$,
the number $C_{j,v}^\ell$
 of
rows in the design matrix 
$J = T_1 \Join \ldots \Join T_m$
that satisfy constraint $\mathcal{L}$, that have
label $\ell$, and
that have value $v$  in
column $j$. \cite{FAQ-AI(1)} gives
a relational algorithm, which we will
call the Row Counting Algorithm,
that computes a $(1+\epsilon)$-approximation for each such
$\widehat C^\ell_{j, v}$ to each $C^\ell_{j, v}$, 
and 
that runs in time $O(\frac{m}{\epsilon^2}(m^3 \log^2(n))^2  (d^2 m n^{h} \log(n)) )$ 

In the second problem the
input is a collection $T_1, \ldots, T_m$ of
tables, a label $\ell \in \{-1, +1\}$,
and an expression in the form of  $\sum_{ j \in [d]} g_j(x_{j})$, where the $g_j$ functions can be
computed in constant time. 
The output
consists of, for each $k \in [0, \log_{1+\epsilon} N ]$, maximum value of $H_k$
such that the number of points 
in the design matrix 
$J = T_1 \Join \ldots \Join T_m$
with label $\ell \in \{-1,1\}$ satisfying the additive
inequality $ \sum_{ j \in [d]} g_j(x_{j}) \geq H_k$ is at least $\lfloor (1+\epsilon)^k \rfloor$.  
\cite{FAQ-AI(1)} gives an algorithm
for this problem, which we will call the
Generalized Row Counting Algorithm,
that runs in time $O(\frac{1}{\epsilon^2}(m^3 \log^2(n))^2  (d^2 m n^{h} \log(n)))$. 
Using the result of the algorithm, for any scalar distance $H$,  it is possible to obtain a row count $\hat{N}(H)$ such that $N(H)/(1+\epsilon) \leq \hat{N}(H) \leq N(H)$, where $N(H)$ is the number of points in the design matrix with label $\ell$ satisfying the inequality $ \sum_{ j \in [d]} g_j(x_{j}) \geq H_k$.

\subsection{Overview of Our Approach}

Recall from the introduction that the difficulty arises
when a $\widehat G_k^+$ is 
approximately equal to $-\widehat G_k^-$. 
In this case, it would seem that by appropriately perturbing one of $L_1^-$ or $L_1^+$ 
by a relatively small amount one could force 
$ G = \widehat G^- + \widehat G^+$ for this perturbed instance.
In which case, if we used $2 \lambda \beta^{(t)} + (\widehat G^- +\widehat G^+)$
as the pseudo-gradient, then it would be the true gradient
for a slightly perturbed instance. 
However, this isn't quite right, as there is an additional issue. If we perturb a point $x_i$, then the sign of $1- y_i \beta x_i$ may change, which means 
this point's contribution to the gradient may discontinuously switch between $0$ and $- y_i x_i$.  To address
this issue, when computing the pseudo-gradient, we use a new instance $X'$ that excludes points that are ``close" to the separating hyperplane $1- y_i \beta x_i = 0$. That is, $X'$ 
excludes every point that can change sides of the hyperplane in  an $\epsilon$-perturbation of each coordinate. This will allow us
to formally conclude that if we used $2 \lambda \beta^{(t)} + (\widehat G^- + \widehat G^+)$,
where $\widehat G^-$ and $\widehat G^+$ 
are defined on $X'$, 
as the pseudo-gradient, then it would be the true gradient
for a slightly perturbed instance.
After the last descent step, 
we choose the final hypothesis to be 
the $\epsilon$-perturbation
of any computed hypothesis $\beta^{(t)}$, $t \in [0, T]$
that minimizes the SVM objective. 

 In the analysis we interpret the sequence $\beta^{(0)}, \beta^{(1)}, \ldots, \beta^{(T)}$ as  solving an online convex optimization problem, and apply known techniques
 from this area.


\subsection{Pseudo-gradient Descent Algorithm}
\label{sec:algorithm_design}

Firstly, in linear time it is straight-forward
to determine if the points in $X$ lie in
$[-1, 1]$, and if not, to rescale so
that they do; This can be accomplished by,
for each feature, dividing all the values of that feature in all of the input tables by maximum absolute value of  that feature. 
The initial hypothesis $\beta^{(0)}$ is the origin. For any vector $v$, let $u = \abs{v}$ be a vector such that its entries are the absolute values of $v$, meaning for all $j$ $u_j = \abs{v_j}$.

\noindent\textbf{Algorithm to Compute the Pseudo-gradient:}
\begin{enumerate}
\item
Run the Row Counting Algorithm to compute,
for each $j \in [d]$ and $v \in D(j)$,
a $(1+\epsilon)$ approximation 
$\widehat C^-_{j, v}$ to $C^-_{j,v}$, which is the number of 
rows in $x\in J$ with negative label,
satisfying $1+\beta^{(t)} \cdot x \geq \epsilon \abs{\beta^{(t)}} \cdot \abs{x}$. 

\item
Run the Row Counting Algorithm to compute,
for each $j \in [d]$ and $v \in D(j)$,
a $(1+\epsilon)$ approximation 
$\widehat C^+_{j, v}$ to $C^+_{j,v}$, which is the number of 
rows in $x\in J$ with positive label,
satisfying $1-\beta^{(t)} \cdot x \geq \epsilon \abs{\beta^{(t)}} \cdot \abs{x}$. 

    \item  For all $k \in [d]$, compute $\widehat G^-_k =  \sum_{v \in D(k) : v < 0} v \; \widehat C_{k,v}^- - \sum_{v \in D(k) : v \geq 0} v \; \widehat C_{k,v}^+$ .

    \item  For all $k \in [d]$, compute $\widehat G^+_k =  \sum_{v \in D(k) : v \geq 0} v \; \widehat C_{k,v}^- - \sum_{v \in D(k) : v < 0} v \; \widehat C_{k,v}^+$.

    \item \label{step:pseudo_gradient} The pseudo-gradient is then
    \begin{align*}
        \widehat G = \frac{\widehat G^- + \widehat G^+}{N} + \lambda \beta^{(t)}
    \end{align*}
\end{enumerate}

\noindent\textbf{Algorithm for a Single Descent Step:}
The next hypothesis $\beta^{(t+1)}$ is 
    \begin{align*}
        \beta^{(t+1)} = \Pi_{\mathcal{K}}(\beta^{(t)} - \eta_{t+1} \widehat G)
    \end{align*} 
  Here $\eta_t = \frac{1}{\lambda \sqrt{dt}}$ and $\Pi_\mathcal{K}(\beta)$ is the projection of $\beta$ onto a hypersphere $\mathcal{K}$ centered at the origin with radius $\frac{\sqrt{d}}{2\lambda}$.
    Note that $\Pi_{\mathcal{K}}(\beta)$ is $\beta$ if $\norm{\beta}_2 \leq \frac{\sqrt{d}}{2\lambda}$ and $\frac{\sqrt{d}}{2 \lambda \norm{\beta}_2} \beta$ otherwise.

\noindent\textbf{Algorithm to Compute 
the Final Hypothesis:} After $T-1$ descent steps, 
the algorithm calls the Generalized Row Counting twice 
for each $t \in [0, T-1]$, with the following inputs:
\begin{itemize}
    \item $\ell = 1$ and additive expression $1-\beta^{(t)} \cdot x_i - \epsilon |\beta^{(t)}| \cdot |x_i| $
    \item $\ell = -1$ and additive expression $1+\beta^{(t)} \cdot x_i - \epsilon |\beta^{(t)}| \cdot |x_i| $
\end{itemize}
Note that both of these expressions are equivalent to $1-y_i\beta^{(t)} \cdot x_i - \epsilon |\beta^{(t)}| \cdot |x_i|$. Let the array $H^+$ be the output for the first call and $H^-$ be the output for the second call. Note that $H^+$ and $H^-$ are monotonically decreasing by the the definition of the Generalized Row Counting algorithm. Let $L^+$ be the largest $k$ such that $H^+_k \geq 0$ and $L^-$ be the largest $k$ such that $H^-_k \geq 0$.
The algorithm then returns 
as its final hypothesis $\widehat \beta$, the hypothesis $\beta^{(\widehat t)}$ where $\widehat t$ is defined by:
\begin{align}
     \widehat t = \argmin_{t \in [T]} \widehat F(\beta^{(t)}, X)\end{align}
     where
\begin{equation}
\begin{aligned}
\label{eqn:fhat}
    \widehat F(\beta^{(t)}, X) 
    = &\frac{1}{N}\left(\sum_{k=0}^{L^+-1} (1+\epsilon)^k (H^+_{k} - H^+_{k+1}) + (1+\epsilon)^{L^+} H^+_{L^+}\right)
    \\
    +&\frac{1}{N} \left(\sum_{k=0}^{L^--1} (1+\epsilon)^k (H^-_{k} - H^-_{k+1}) + (1+\epsilon)^{L^-} H^-_{L^-}\right)
    + \lambda \norm{\beta^{(t)}}^2_2  
\end{aligned}
\end{equation}
Note that the values $L^-$, $L^+$, $H^+$ and $H^-$ 
in the definition of $\widehat F$, in equation (\ref{eqn:fhat}), all 
depend upon $t$, which we suppressed to make the notation
somewhat less ugly.

\section{Algorithm Analysis}
\label{sect:analysis}

In subsection \ref{subsec:perbanalysis} we prove
Theorem \ref{thm:svm} which bounds the convergence
of our project pseudo-gradient descent algorithm
in a rather nonstandard way by applying known
results on online convex optimization~\cite{OCO,Hazan16}.  
In subsection \ref{subsect:stability} we
introduce our definition of stability and
then prove Theorem \ref{thm:main}.

\subsection{Perturbation Analysis}
\label{subsec:perbanalysis}

Before stating Theorem \ref{thm:svm} we need some definitions.

\begin{definition}~
\label{defn:perturbation}
\begin{itemize}
    \item 
A point $p$ is an $\epsilon$-perturbation of point $q$ if every component of $p$ is within $(1 + \epsilon)$ factor of the corresponding component of $q$. Meaning in each dimension $j$ we have $(1-\epsilon)q \leq p \leq (1+\epsilon) q$
\item
A point set $X_a$ is  an $\epsilon$-perturbation of a point set $X_b$ if there
is a bijection between $X_a$ and $X_b$ such that
every point in $X_a$ is an $\epsilon$-perturbation of its corresponding point in $X_b$.
\item
Let 
$\beta^* =\argmin_\beta F(\beta, X)$ to be the optimal solution
at $X$.
\item For any $\epsilon$-perturbation $X_a$ of $X$, define
$\beta^*_a =\argmin_\beta F(\beta, X_a)$ to be the optimal solution
at $X_a$.
\item
For a given hypothesis $\beta$, we call a point $x$ with label $y$ \emph{close} if there is some $\epsilon$-perturbation $x'$ of $x$ such that $1-y \beta x'  < 0$; otherwise it is called \emph{far}. In other words, a point $x$ with label $y$ is close if $1-y \beta \cdot x < \epsilon \abs{\beta} \cdot \abs{x}$
\end{itemize}
\end{definition}

\begin{theorem}
	\label{thm:svm} 
Assume our projected pseudo-gradient descent algorithm
ran for $T-1$ descent steps. 
Then for all hypotheses $\beta \in \mathbb{R}^d$ there
 exist $\epsilon$-perturbations $X_a$ and $X_b$ of 
 $X$ such that
\begin{align*}
    F(\widehat \beta,X_a) \leq (1+\epsilon) F(\beta,X_b) + \frac{ 32d^{3/2}}{\lambda \sqrt{T}}
\end{align*}
\end{theorem}

To prove Theorem \ref{thm:svm}, our main tool is a
result from the online convex optimization literature~\cite{OCO,Hazan16}.

\begin{theorem}\cite{OCO,Hazan16} 
	\label{thm:oco}
	Let $g_1, g_2, ..., g_T: \mathbb{R}^n \rightarrow \mathbb{R}$ be $G$-Lipschitz functions over
	a convex region $\mathcal{K}$, i.e., $|| \nabla g_t(\beta) || \leq G$ for all $\beta \in \mathcal{K}$ and all $t$. Then, starting at point $\beta^{(0)} \in \mathbb{R}^n$ and using the update rule of $\beta^{(t)} \leftarrow \Pi_\mathcal{K} \left( \beta^{(t-1)} - \eta_t \nabla g_{t-1}( \beta^{(t-1)})\right)$, with $\eta = \frac{D}{G \sqrt t}$ for $T-1$ steps, we have
\begin{align}	
	\label{eqn:oco}
	\frac{1}{T} \sum_{t = 0}^{T-1} g_t( \beta^{(t)}) \leq \frac{1}{T} \sum_{t = 0}^{T-1} g_t( \beta^*) + \frac{2DG}{\sqrt {T}}
\end{align}
\noindent	 for all $\beta^*$ with $|| \beta^{(0)} - \beta^*|| \leq D$. 
\end{theorem}

To apply this Theorem \ref{thm:oco}, we set $g_t =    F(\beta^{(t)},X^{(t)}, Y)$, where $X^{(t)}$ is 
an $\epsilon$-perturbation of $X$, such
that the pseudo-gradient at $X$ is equal to the 
true gradient at $X^{(t)}$. 
We establish the existence of
$X^{(t)}$ in Lemma \ref{prop:g-of-something}.
Thus our projected pseudo-gradient descent algorithm updates the hypothesis  exactly the same as stated in Theorem \ref{thm:svm} (assuming that we
use the same upper bounds on $D$ and $G$).
Then in definition \ref{defn:Z} we identify  the
$\epsilon$-permutation $Z$ that minimizes $F(\beta, Z)$,
and then in Lemma \ref{lemma:hatnohat} bound
the relative error between $\widehat F(\beta, X)$ 
and $F(\beta, Z)$. Finally this will allow use
in Lemma \ref{lemma:Xb} and Lemma \ref{lemma:Xa}
we show the existence of $X_b$ and $X_a$, respectively,
that will allow us to conclude the proof of Theorem \ref{thm:svm}.

\begin{lemma}
	\label{prop:g-of-something}
	In every descent step $t$, the computed pseudo-gradient
	$\widehat{G}$ is the exact gradient of $F(\beta^{(t)}, X^{(t)})$ for some point set $X^{(t)}$ that is an $\epsilon$-perturbation of $X$. 
\end{lemma}
\begin{proof}
To prove the claim, we show how to find a desired $X^{(t)}$ -- this is only for the sake of the proof and the algorithm doesn't need to know $X^{(t)}$. 
We call any point $x$ with label $y$ ``far'' if it satisfies the inequality 
\begin{align}
\label{equality:far_close}
    1-y \beta^{(t)} \cdot x \geq \epsilon  \abs{\beta^{(t)}} \cdot \abs{x}
\end{align}, otherwise we call the point ``close''. 
Note that for a far point there is no $\epsilon$-perturbation to make the derivative of the loss function $0$. 
That is, for any point $x$ with label $y$, if $1-y \beta \cdot x \geq \epsilon \sum_{j \in [d]} \abs{\beta_{j}}  \abs{x_j}$, then we have $1-  y \beta x' \geq 0$ for any $x'$ that is $\epsilon$-perturbation of $x$. To see this, note that we have $1 -  y \beta x' = 1 - \sum_{k = 1}^d  \beta_{k} x'_{k} \geq 1 - \sum_{k = 1}^d  \left(\beta_{k} x_{k} + \abs{\beta_{k}} \abs{x_{k}}\right) \geq 0$ because of $x'$ being $\epsilon$-perturbation of $x$. 
On the other hand, for all the close points there exists a perturbation $x'$ such that $1-y\beta^{(t)} \cdot x' < 0$. We first perturb all of the close points such that they don't have any effect on the gradient. 

Next, we need to show a perturbation of the far points for which the $\widehat{G}$ is the gradient of the loss function. Let $X^+_f$ and $X^-_f$ be the set of far points with positive and negative labels. Let $X_f = X^+_f \cup X^-_f$. We show the perturbation for each dimension $k$ separately. Based on definition of $\widehat{G}^+_k$ and $\widehat{G}^-_k$ we have:
\begin{align*}
    \widehat{G}^+_k + \widehat{G}^-_k
    &= \sum_{v \in D(k)} v \; \widehat C_{k,v}^- - \sum_{v \in D(k)} v \; \widehat C_{k,v}^+
    \\
    &= \sum_{v \in D(k)} v \; (1\pm \epsilon) C_{k,v}^- - \sum_{v \in D(k)} v \; (1 \pm \epsilon) C_{k,v}^+
\end{align*}
Note that $C^+_{k,v}$ is the number of points in $X^+_f$ with value $v$ in dimension $k$. Therefore,
\begin{align*}
    \widehat{G}^+_k + \widehat{G}^-_k 
    &= \sum_{v \in D(k)} v \; (1\pm \epsilon) C_{k,v}^- - \sum_{v \in D(k)} v \; (1 \pm \epsilon) C_{k,v}^+
    \\
    &= \sum_{x_i \in X^-_f} (1 \pm \epsilon) x_{i,k} - \sum_{x_i \in X^+_f} (1 \pm \epsilon) x_{i,k}
    \\
    &= - \sum_{x_i \in X_f} (1 \pm \epsilon) y_i x_{i,k} 
\end{align*}
where the last term is $N \frac{\partial L(\beta^{(t)},X^{(t)})}{\partial \beta^{(t)}_k}$ where $X^{(t)}$ an $\epsilon$-perturbation of $X$.
\end{proof}

\begin{definition}
\label{defn:Z}
Let $Z^{(t)}$ be an $\epsilon$-perturbation of $X$ such that for all $z_i \in Z^{(t)}$ and for all dimensions $k$
\begin{align*}
    z_{i,k} = \begin{cases}
    (1-\epsilon) x_{i,k} & y_i \beta^{(t)}_{k} \geq 0
    \\
    (1+\epsilon) x_{i,k} & y_i \beta^{(t)}_{k} < 0
    \end{cases}
\end{align*}
Note that this $\epsilon$-perturbation minimizes $F(\beta^{(t)}, Z^{(t)})$.

\end{definition}

\begin{lemma}
\label{lemma:hatnohat}
  $\frac{1}{1+\epsilon} F(\beta^{(t)}, Z^{(t)}) \leq\widehat{F}(\beta^{(t)}, X) \leq F(\beta^{(t)}, Z^{(t)})$.
  \end{lemma}

\begin{proof}
Consider a value $t$ and let $N^+(\tau) = \abs{\{ x_i \suchthat y_i = +1 \allowbreak \text{ and } \allowbreak 1- \allowbreak \beta^{(t)} \cdot x_i - \epsilon \abs{\beta^{(t)}} \cdot \abs{x_i} \geq \tau \}}$, and $N^-(\tau) = \abs{\{ x_i \suchthat y_i = -1 \allowbreak \text{ and } \allowbreak 1+ \allowbreak \beta^{(t)} \cdot x_i - \epsilon \abs{\beta^{(t)}} \cdot \abs{x_i} \geq \tau \}}$.

Before proving the lemma we prove the following claim: $F(\beta^{(t)}, Z^{(t)}) = \frac{1}{N}\int_{\tau = 0}^\infty N^+(\tau) \text{d} \tau + \frac{1}{N}\int_{\tau = 0}^\infty N^-(\tau) \text{d} \tau + \lambda \norm{\beta^{(t)}}^2$.
 
Note that based on the definition of $Z^{(t)}$ it is the case that $1-y_i \beta^{(t)} \cdot z_i = 1-y_i \beta^{(t)} \cdot x_i - \epsilon \abs{\beta^{(t)}} \cdot \abs{x_i}$; therefore, $N^+(\tau) = \abs{\{ y_i = +1 \in Z^{(t)} \text{ and } 1-y_i \beta^{(t)} \cdot z_i \geq \tau \}}$ and $N^-(\tau) = \abs{\{ y_i = -1 \in Z^{(t)} \text{ and } 1-y_i \beta^{(t)} \cdot z_i \geq \tau \}}$.  
Hence, 
\begin{align*}
    L(\beta^{(t)}, Z^{(t)}) 
    &= \frac{1}{N} \sum_i \max(0,1-y_i \beta \cdot z_i)
    = \frac{1}{N} \sum_{i: 1-y_i \beta \cdot z_i \geq 0} 1-y_i \beta \cdot z_i
    \\
    &= \frac{1}{N} \sum_{i: 1-y_i \beta \cdot z_i \geq 0} \int_{\tau = 0}^{1-y_i \beta \cdot z_i} \text{d} \tau
    = \frac{1}{N} \int_{\tau = 0}^{\infty}  \sum_{i: 1-y_i \beta \cdot z_i \geq \tau} \text{d} \tau 
    \\
    &= \frac{1}{N}\int_{\tau = 0}^\infty (N^+(\tau)+N^-(\tau)) \text{d} \tau
\end{align*}
Therefore, 
\begin{align}
    \label{eq:integral}
    F(\beta^{(t)}, Z^{(t)}) = \frac{1}{N}\int_{\tau = 0}^\infty N^+(\tau) \text{d} \tau + \frac{1}{N}\int_{\tau = 0}^\infty N^-(\tau) \text{d} \tau + \lambda \norm{\beta^{(t)}}^2
\end{align}
The number of points with label $\ell$ satisfying $1- \ell \beta^{(t)} \cdot  x_i - \epsilon \abs{\beta^{(t)}} \cdot \abs{x_i} \geq \tau$ for any $\tau \in [H^\ell_{k}, H^\ell_{k+1})$ is in the range $[\lfloor (1+\epsilon)^k \rfloor,\lfloor (1+\epsilon)^{(k+1)} \rfloor)$. Therefore, the claim follows by replacing $N^+(\tau)$ in Equation \eqref{eq:integral} with $(1+\epsilon)^k$ for all the values of $\tau \in [H^+_{k}, H^+_{k+1})$ and replacing $N^-(\tau)$ in $\eqref{eq:integral}$ with $(1+\epsilon)^k$ for all the values of $\tau \in [H^-_{k}, H^-_{k+1})$.

\end{proof}

\begin{lemma}
\label{lemma:Xb} 
	For all hypothesis $\beta$, there exists an $\epsilon$-perturbation
	$X_b$ of $X$ such that
\begin{align*}
	\min_{s} F(\beta^{(s)},Z^{(s)}) \leq F(\beta,X_b) + \frac{2DG}{\sqrt{T}}
\end{align*}
\end{lemma}

\begin{proof} 
By Theorem \ref{thm:oco}
\begin{align}
\label{eqn:pert1}
    \frac{1}{T}\sum_{t=0}^{T-1} F(\beta^{(t)}, X^{(t)}) \leq \frac{1}{T}\sum_{t=0}^{T-1} F(\beta,X^{(t)}) + \frac{2DG}{\sqrt{T}}
\end{align} 
Then
\begin{align}
\label{eqn:pert2}
  \min_{s} F(\beta^{(s)},Z^{(s)}) \leq \frac{1}{T}\sum_{t=0}^{T-1} F(\beta^{(t)}, Z^{(t)}) \leq \frac{1}{T}\sum_{t=0}^{T-1} F(\beta^{(t)}, X^{(t)}).
\end{align}
The first inequality follows since the minimum is less
than the average, and the second inequality follows
from the definition of $Z^{(t)}$.
Let $u = \argmax_t F(\beta, X^{(t)})$, and $X_b = X^{(u)}$. 
Then 
\begin{align}
\label{eqn:pert3}
    \frac{1}{T}\sum_{t=0}^{T-1} F(\beta,X^{(t)}) \leq \max_t F(\beta,X^{(t)}) = F(\beta, X_b)
\end{align}
Thus, combining lines (\ref{eqn:pert1}), (\ref{eqn:pert2}) and (\ref{eqn:pert3}) we can conclude that:
\begin{align}
	\label{eqn:60}
	\min_{s} F(\beta^{(s)},Z^{(s)}) \leq F(\beta,X_b) + \frac{2DG}{\sqrt{T}}
\end{align}
\end{proof}

\begin{lemma}
\label{lemma:Xa} 
	There exists an $\epsilon$-perturbation
	$X_a$ of $X$ such that
	\begin{align*}
	F(\widehat \beta, X_a) \le (1+\epsilon) \min_{s} F(\beta^{(s)},Z^{(s)}) 
		\end{align*}
\end{lemma}

\begin{proof}
Let $X_a = Z^{(\widehat t)}$ where 
\begin{align*}
   F(\widehat \beta , X_a)  & \le (1+\epsilon) \widehat F(\widehat \beta , X) & \text{By Lemma \ref{lemma:hatnohat}}\\  
   &= (1+\epsilon) \min_s \widehat F( \beta^{(s)}, X) & \text{By definition of } \widehat \beta \\
   &\le  (1+\epsilon) \min_s F( \beta^{(s)}, Z^{(s)}) & \text{By Lemma \ref{lemma:hatnohat}}
\end{align*}

\end{proof}

\subsection{Stability Analysis}
\label{subsect:stability}

Our formal definition of stability, which
we give in Definition \ref{defn:stability}
while not unnatural, is surely not the first
natural formalization that one would think of. 
Our formal definition  
is more or less forced on us, which leads to the type
of non-traditional approximation achieved in Theorem \ref{thm:svm}.

\begin{definition}~
\label{defn:stability}
An SVM instance $X$   is $(\alpha, \delta, \gamma)$-stable for $\delta \leq 1$ if for all
 $X_a$ and $X_b$ that are $\alpha$-perturbations of $X$ it is the case that:
\begin{itemize}
    \item 
 $\beta^*_a $ is a $(1+\delta)$ approximation to the optimal
objective value at $X_b$, that is, 
$F(\beta^*_a, X_b) \le (1+\delta) \min_\beta F(\beta, X_b)$. 
    \item 
If 
$\beta_a$ is  $(1+2 \delta)$ approximation to the optimal SVM objective value at $X_a$   then $\beta_a$ is a $(1+\gamma)$ approximation to the optimal SVM objective value at $X_b$. 
That is if $F(\beta_a, X_a) \le (1+2 \delta) \min_\beta F(\beta, X_a)$
then
$F(\beta_a, X_b) \le (1+ \gamma) \min_\beta F(\beta, X_b)$
\end{itemize}
\end{definition}

\begin{proof}[Proof of Theorem \ref{thm:main}]  
 Let $\epsilon \leq \min(\delta/8,\alpha)$.
    \begin{align*}
 F(\widehat \beta,X_a)  
    &\leq (1+\epsilon) F(\beta^*_a,X_b) + \frac{ 32d^{3/2}}{\lambda \sqrt{T}}&   X_a \text { and } X_b \text{ come from Theorem } \ref{thm:svm}  \\
    &=    (1+\epsilon) (1+\delta) F(\beta^*_a, X_a)  + \frac{32d^{3/2} }{\lambda \sqrt{T}} &   \text{By definition of stability}  \\
     &=    (1+\epsilon) (1+\delta) F(\beta^*_a,X_a)  + \frac{\delta}{8} F(\widehat \beta, X_a) &   \text{By definition of } T \\
      &\le   \frac{(1+\delta)(1+\epsilon)}{1-\delta/8} F(\beta^*_a,X_a)   &  \text{By algebra}  \\
      & \le (1+2 \delta)F(\beta^*_a,X_a) & \text{ by definition of } \epsilon
\end{align*}
Finally since $\widehat \beta$ is $(1+ 2 \delta)$ approximate  solution at $X_a$, by the definition of stability, 
$\widehat \beta$ is a $(1+\gamma) $ approximate solution at $X$.
\end{proof}

\bibliography{arxive}

\newpage

\appendix

\section{Analysis of Gradient Descent for SVM} 
\label{sect:gradientdescentanalysis}

Theorem \ref{thm:gradient} and Corollary \ref{corollary:gradient} give bounds on the
 number of iterations on  projected gradient 
 descent to reach solutions with bounded absolute error
 and bounded relative error, respectively.

\begin{theorem}\cite{OCO,Hazan16}
\label{thm:gradient}
Let $\mathcal{K}$ be a convex body and $F$ be a function such that 
$\norm{\nabla F(\beta)}_2 \le G$ for $\beta \in \mathcal{K}$. 
Let $\beta^* = \argmin_{\beta \in \mathcal{K}}  F(\beta)$ be the optimal solution.
Let $D$ be an upper bound on  $\norm{\beta^{(0)} - \beta^*}_2$,
the 2-norm distance from the initial candidate solution to the optimal solution.
Let $\widehat{\beta}_s = \frac{1}{s}\sum_{t=0}^{s-1} \beta^{(t)}$. 
Let $\eta_t = \frac{D}{G\sqrt{t}}$. Then
after $T-1$ iterations of projected gradient descent, it must be the case that
\begin{align*}
    F(\widehat{\beta}_{T})-F(\beta^*) \leq \frac{2DG}{\sqrt{T}}
\end{align*}
\end{theorem}

\begin{corollary}
\label{corollary:gradient}
Adopting the assumptions from Theorem \ref{thm:gradient}, if $T \ge \left( \frac{4  DG}{\epsilon F(\widehat{\beta}_{T})} \right)^2$  then
\begin{align*}
    F(\widehat{\beta}_{T})\leq  (1+\epsilon) F(\beta^*)  
\end{align*}
That is, projected gradient descent achieves relative error $\epsilon$.
\end{corollary}

The gradient of SVM objective $F$ is
\begin{align*}
\nabla F = 2 \lambda \beta - \frac{1}{N} y_i \sum_{i \in \mathcal{L}} x_i
\end{align*}
where $\mathcal{L}$ is the collection $\{ i \mid \beta x_i \le 1 \}$ of indices $i$ where $x_i$ is currently contributing to the objective. Note that in this hinge loss function, the gradient for the points on the hyperplane $1-\beta x=0$ does not exist, since the gradient is not continuous at this point. In our formulation we have used the sub-gradient for the points on $1-\beta x=0$, meaning for a $\beta$ on the hyperplane $1-\beta x=0$, we have used the limit of the gradient of the points that $1-\beta' x>0$ when $\beta'$ goes to $\beta$. For all the points that $1-\beta' x>0$, the gradient is $x$; therefore, the limit is also $x$.

Assume $\beta^{(0)}$ is the origin and adopt
the assumptions of Theorem \ref{thm:gradient}.
Then $\nabla F (\beta^*) = 0$ implies for any dimension $j$
\begin{align*}
  \abs{\beta^*_j} = \abs{\frac{1}{2 N  \lambda}\sum_{i \in \mathcal{L}} x_{ij}} \leq \frac{1}{  2 \lambda}
\end{align*} where the additional subscript of $j$ 
refers to dimension $j$.
And thus
\begin{align*}
    \norm{\beta^{(0)} - \beta^*}_2 \le \norm{\beta^*}_2 \le \sqrt{d} \max_{j \in [d]} \abs{\beta^*_j} \le \frac{\sqrt{d}}{2\lambda} 
\end{align*}
Thus let us define our convex body $\mathcal{K}$ to
be the hypersphere with radius $\frac{\sqrt{d}}{2\lambda}$
centered at the origin. Thus for $\beta \in \mathcal{K}$, 
\begin{align*}
\norm{\nabla F(\beta)}_2 &= \sqrt{\sum_{j \in [d]} \left( 2 \lambda \beta_j -     \frac{1}{N}\sum_{i \in \mathcal{L}} x_{ij} \right)^2} &\\
& \le \sqrt{\sum_{j \in [d]} 4 (\lambda \beta_j )^2    + 2 \left(\frac{1}{N} \sum_{i \in \mathcal{L}} x_{ij} \right)^2 } &\text{Since } (a-b)^2 \le 2a^2 + 2b^2    \\
&\le  2 \lambda \sqrt{\sum_{j \in [d]}   \beta_j^2 }     + \sqrt{2}  \frac{1}{N} \sum_{j \in [d]} \sum_{i \in \mathcal{L}} |x_{ij}| & \text{Since } \sqrt{\sum_{i} a_i^2} \le \sum_{i} |a_i| \\
&\le \sqrt{d} + \sqrt{2} d &  \\
&\le 4 d  &
\end{align*}


\begin{theorem}
\label{thm:gradientsvm}
Let the convex body $\mathcal{K}$
be the hypersphere with radius $\frac{\sqrt{d}}{2\lambda}$
centered at the origin. Let $F(\beta)$ be the SVM objective function.
Let $\beta^* = \argmin_\beta  F(\beta)$ be the optimal solution.
Let $\widehat{\beta}_s = \frac{1}{s}\sum_{t=0}^{s-1} \beta^{(t)}$. 
Let $\eta_t = \frac{1}{8 \lambda \sqrt{dt}}$. Then
after $T-1$ iterations of projected gradient descent, it must be the case that
\begin{align*}
    F(\widehat{\beta}_{T})-F(\beta^*) \leq \frac{ 4 d^{3/2} }{\lambda \sqrt{T}}
\end{align*}
\end{theorem}

Theorem \ref{corollary:gradientsvm} then 
follows by a straightforward application of 
Theorem \ref{thm:gradientsvm}.

\section{Background}
\label{app:background}

\subsection{Fractional edge cover number and output size bounds}
\label{app:agm}
In what follows, we consider a conjunctive query $Q$ over a relational database instance $I$.
We use $n$ to denote the size of the largest input relation in $Q$.
We also use $Q(I)$ to denote the output and $|Q(I)|$ to denote its size.
We use the query $Q$ and its hypergraph $\calH$ interchangeably.
\bdefn[Fractional edge cover number $\rho^*$]
Let $\calH=(\calV,\calE)$ be a hypergraph (of some query $Q$). Let $B\subseteq\calV$ be any subset
of vertices. 
A {\em fractional edge cover} of $B$ using edges in $\calH$ is a feasible
solution $\vec\lambda =(\lambda_S)_{S\in\calE}$ to the following linear
program:
\begin{eqnarray*}
	\min && \sum_{S\in\calE} \lambda_S\\
	\text{s.t.}&& \sum_{S : v \in S} \lambda_S \geq 1, \ \ \forall v \in B\\
	&& \lambda_S \geq 0, \ \ \forall S\in \calE.
\end{eqnarray*}
The optimal objective value of the above linear program is called
the {\em fractional edge cover number} of $B$ in $\calH$ and is denoted by $\rho^*_\calH(B)$.
When $\calH$ is clear from the context, we drop the subscript $\calH$ and use $\rho^*(B)$.

Given a conjunctive query $Q$, the fractional edge cover number of $Q$ is $\rho^*_\calH(\calV)$
where $\calH=(\calV,\calE)$ is the hypergraph of $Q$.
\edefn

\begin{theorem}[AGM-bound~\cite{AGMBound,GM06}]
Given a full conjunctive query $Q$ over a relational database instance $I$,
the output size is bounded by
\[|Q(I)| \leq n^{\rho^*},\]
where $\rho^*$ is the fractional edge cover number of $Q$.
\label{thm:agm-upperbound}
\end{theorem}

\begin{theorem}[AGM-bound is tight~\cite{AGMBound,GM06}]
Given a full conjunctive query $Q$ and a non-negative number $n$,
there exists a database instance $I$ whose relation sizes are upper-bounded by $n$ and satisfies
\[|Q(I)| =\Theta(n^{\rho^*}).\]
\label{thm:agm-lowerbound}
\end{theorem}

\emph{Worst-case optimal join algorithms}~\cite{LFTJ,Ngo2012wcoj,skew} can be used to answer any full conjunctive query $Q$
in time
\begin{equation}
O(|\calV|\cdot|\calE|\cdot n^{\rho^*}\cdot \log n).
\label{eqn:runtime:lftj}
\end{equation}

\subsection{Tree decompositions, acyclicity, and width parameters}
\label{app:td}
\bdefn[Tree decomposition]
\label{defn:TD}
Let $\calH = (\calV, \calE)$ be a hypergraph.
A {\em tree decomposition} of $\calH$ is a pair $(T, \chi)$
where $T = (V(T), E(T))$ is a tree and $\chi : V(T) \to 2^{\calV}$ assigns to
each node of the tree $T$ a subset of vertices of $\calH$.
The sets $\chi(t)$, $t\in V(T)$, are called the {\em bags} of the 
tree decomposition.  There are two properties the bags must satisfy
\bi
\item[(a)] For any hyperedge $F \in \calE$, there is a bag $\chi(t)$, $t\in
V(T)$, such that $F\subseteq \chi(t)$.
\item[(b)] For any vertex $v \in \calV$, the set 
$\{ t \suchthat t \in V(T), v \in \chi(t) \}$ is not empty and forms a 
connected subtree of $T$.
\ei
\edefn

\bdefn[acyclicity]\label{defn:alpha-acyclic-td}
A hypergraph $\calH = (\calV, \calE)$ is {\em acyclic} iff
there exists a tree decomposition 
$(T, \chi)$ in which every bag $\chi(t)$ is a hyperedge of $\calH$.
\edefn

When $\calH$ represents a join query, the tree $T$ in the above 
definition
is also called the {\em join tree} of the query. 
A query is acyclic if and only if its hypergraph is acyclic.

For non-acyclic queries, we often need a measure of how ``close'' a query is to being acyclic. To that end, we use \emph{width} notions of a query.

\bdefn[$g$-width of a hypergraph: a generic width notion~\cite{adler:dissertation}]
\label{defn:g-width}
Let $\calH=(\calV,\calE)$ be a hypergraph, and
$g : 2^\calV \to \mathbb R^+$ be a function that assigns a non-negative
real number to each subset of $\calV$.
The {\em $g$-width} of a tree decomposition $(T, \chi)$ of $\calH$ is 
$\max_{t\in V(T)} g(\chi(t))$.
The {\em $g$-width of $\calH$} is the {\em minimum} $g$-width
over all tree decompositions of $\calH$.
(Note that the $g$-width of a hypergraph is a {\em Minimax} function.)
\edefn

\bdefn[{\em Treewidth} and {\em fractional hypertree width} are special cases of {\em $g$-width}]
Let $s$ be the following function:
$s(B) = |B|-1$, $\forall V \subseteq \calV$.
Then the {\em treewidth} of a hypergraph $\calH$, denoted by
$\tw(\calH)$, is exactly its $s$-width, and
the {\em fractional hypertree width} of a hypergraph $\calH$,
denoted by $\fhtw(\calH)$, is the $\rho^*$-width of $\calH$.
\edefn

From the above definitions, $\fhtw(\calH)\geq 1$ for any hypergraph $\calH$.
Moreover, $\fhtw(\calH)=1$ if and only if $\calH$ is acyclic.

\subsection{Algebraic Structures}

In this section, we define some of the algebraic structures used in the paper.   First, we discuss the definition of a monoid.  A monoid is a semi-group with an identity element.   Formally, it is the following. 

\begin{definition}
Fix a set $S$ and let $\oplus$ be a binary operator $S \times S \rightarrow S$.  The set $S$ with $\oplus$ is a monoid if (1) the operator satisfies associativity; that is, $ (a \oplus b) \oplus c = a \oplus (b \oplus c)$ for all $a,b,c \in S$ and (2) there is identity element $e \in S$ such that for all $a \in S$, it is the case that $e \oplus a = a \oplus e = e$.

A commutative monoid is a moniod where the operator $\oplus$ is commutative.   That is $a \oplus b = b \oplus a$ for all $a,b \in S$.
\end{definition}

Next, we define a semiring.

\begin{definition}
A semiring is a set $R$ with two operators $\oplus$ and $\otimes$.  The $\oplus$ operator is referred to as addition and the $\otimes$ is referred to as multiplication. This is a semiring if, 

\begin{enumerate}
\item it is the case that $R$ and $\oplus$ are a commutative monoid with $0$ as the identity.
\item $R$ and $\otimes$ is a monoid with identity $1$.
\item the multiplication distributes over addition.  That is for all $a,b,c \in R$ it is the case that $a \otimes (b \oplus c) = (a\otimes b) \oplus (a \otimes c)$ and $ (b \oplus c) \otimes a  = (b\otimes a) \oplus (c \otimes a)$.
\item the $0$ element annihilates $R$. That is, $a \otimes 0 = 0$ and $0\otimes a = 0$ for all $a \in R$.
\end{enumerate}

A commutative  semiring is a semiring where the multiplication is commutative.  That is, $a \otimes b = b \otimes a$ for all $a,b \in S$.
\end{definition}

\subsection{FAQ-AI Query}

The input to FAQ-AI problem consists of three components:
\begin{itemize}
\item
A collection of relational tables $T_1, \ldots T_m$ with real-valued entries. Let
$J = T_1 \Join T_2 \Join \dots \Join T_m$ be the
design matrix that arises from the inner join of the tables. 
Let $n$ be an upper bound on the number of rows in any table $T_i$,
let $N$ be the number of rows in $J$, and let $d$ be the number of columns in $J$. 
\item
An FAQ $Q(J)$ that is either a SumProd query or a SumSum query. 
We define a SumSum query to be a query of the form:
\begin{align*}
   Q(J) = \bigoplus_{x \in J} \bigoplus_{i=1}^d F_{i}(x_i)
\end{align*}
where $(R, \oplus, I_0)$ is a commutative monoid over the arbitrary set $R$ with identity $I_0$.
We define a SumProd query to be a query of the form: 
\begin{align*}
   Q(J) =  \bigoplus_{x \in J} \bigotimes_{i=1}^d F_{i}(x_i)
\end{align*}
where $(R, \oplus, \otimes, I_0, I_1)$ is a commutative semiring over the arbitrary set $R$ with additive identity $I_0$
and multiplicative identity $I_1$. 
In each case, $x_i$ is the entry in
column $i$ of $x$, and
 $F_{i}$ is an arbitrary function with range $R$.
\item
A collection ${\mathcal L} = \{ (G_1, L_1), \ldots (G_b, L_b) \}$ 
where $G_i$ is a collection $\{g_{i,1}, g_{i,2}, \ldots g_{i, d} \}$ of $d$  functions that map the column domains to the reals, 
 and each $L_i$ is a scalar. 
\end{itemize}
FAQ-AI($k$) is a special case of FAQ-AI when
the cardinality of $\mathcal L$ is at most $k$.

The output for the FAQ-AI problem is the result
of the query on the subset of the design matrix 
that satisfies the additive inequalities. 
That is, the output for the FAQ-AI instance with a SumSum query is:
\begin{align}
\label{equality:rfaqli-sumsum}
      Q({\mathcal L}(J)) = \bigoplus_{x \in {\mathcal L}(J)} \bigoplus_{i=1}^d F_{i}(x_i)
\end{align}
And the output for the FAQ-AI instance with a SumProd query is:
\begin{align}
\label{equality:rfaqli-sumprod}
    Q({\mathcal L}(J)) = \bigoplus_{x \in {\mathcal L}(J)} \bigotimes_{i=1}^d F_{i}(x_i)
\end{align}
Here ${\mathcal L}(J)$ is the set of tuples $x \in J$ that satisfy all the additive inequalities in
$\mathcal L$, that is  for all $i \in [1, b]$, $\sum_{j=1}^d g_{i,j}(x_j) \le L_i$, where $x_j$ is the 
value of coordinate $j$ of $x$.

We now illustrate how some of the SVM related problems  can be reduced to FAQ-AI(1). 
First consider the problem of counting the number of negatively labeled points correctly classified by a linear separator.  Here each 
row $x$ of the design matrix $J$ conceptually consists of a point in $\mathbb{R}^{d-1}$, whose
coordinates are specified by the first $d-1$ columns in $J$, 
and a label in $\{1,-1\}$ in column $d$. Let  the  linear separator be defined by $\beta \in \mathbb{R}^{d-1}$.  A negatively labeled point $x$ is correctly classified if $\sum_{i=1}^{d-1} \beta_i x_i \leq 0$.  The number of such points can be counted using SumProd query with
one additive inequality as follows:
$\oplus$ is addition, 
$\otimes$ is multiplication, 
$F_i(x_i) = 1$ for all $i \in [d-1]$, 
$F_d(x_d) = 1$ if $x_d = -1$, and $F_d(x_d) = 0$ otherwise,
$g_{1,j}(x_j)= \beta_j x_j$ for $j \in [d-1]$, 
$g_{1,d}(x_d) = 0$, and
$L_1 =0$.
Next, consider the problem of finding the minimum distance to the linear separator of a correctly classified negatively labeled point. 
This distance can be computed using a SumProd query with one additive inequality as follows:
$\oplus$ is the binary minimum operator, 
$\otimes$ is addition, 
$F_i(x_i) = \beta_i x_i$ for all $i \in [d-1]$, 
$F_d(x_d) = 1$ if $x_d = -1$, and $F_d(x_d) = 0$ otherwise,
$g_{1,j}(x_j)= \beta_j x_j$ for $j \in [d-1]$, 
$g_{1,d}(x_d) = 0$, and
$L_1 =0$.

\end{document}